\documentclass{fundam}

\usepackage[ruled,lined]{algorithm2e}% provides Algorithm environment
\usepackage{graphicx}% allows for inclusion of graphic files (figures)
\usepackage{tikz}
\usetikzlibrary{automata, positioning, arrows}

\usepackage{amssymb,amsmath,amsfonts}  %ams
\usepackage{dsfont}
\usepackage{amsfonts}

 \usepackage{mymacros}

\usepackage{hyperref}

\begin{document}

%%%%%%%  parameters to be filled in by copy-editor  %%%%%%%%%%

\setcounter{page}{239}
\publyear{24}
\papernumber{2181}
\volume{191}
\issue{3-4}

\finalVersionForARXIV
%\finalVersionForIOS

%%%%%%%%%%%%%%%%%%%%%%%%%%%%%%%%%%%%%%

\title{Nonatomic Non-Cooperative Neighbourhood Balancing Games}

\author{David Auger\thanks{Address for correspondence:   Universit\'e Paris-Saclay, UVSQ, DAVID,
                     78035, Versailles, France.}\thanks{This work was supported by a public grant as part of the
                     Investissement d'avenir project, reference ANR-11-LABX-0056-LMH, LabEx LMH.}
 \\
Universit\'e Paris-Saclay, UVSQ, DAVID\\
78035, Versailles, France\\
david.auger{@}uvsq.fr
\and Johanne Cohen\thanks{idem}\\
Universit\'e Paris-Saclay, CNRS, LISN  \\
 91405, Orsay, France \\
Johanne.Cohen{@}universite-paris-saclay.fr
\and Antoine Lobstein\\
Universit\'e Paris-Saclay, CNRS, LISN  \\
 91405, Orsay, France \\
Antoine.Lobstein{@}lri.fr
}

\maketitle

\runninghead{D.~Auger et al.}{Neighbourhood Balancing Games}

\begin{abstract}

  We introduce a game where players selfishly choose a resource and endure a cost depending on the number of players choosing nearby resources. We model the influences among resources by a weighted graph, directed or not.
  These games are generalizations of well-known games like Wardrop and congestion games. We study the conditions of equilibria existence and their efficiency if they exist. We conclude with studies of games whose influences among resources can be modelled by simple graphs.

\end{abstract}


\section{Introduction}

Summer is at its peak, and you spend the day at the beach with your whole family. Unfortunately, you were not the only one with this idea, and the beach is crowded with optimization specialists. Where will you put your towels?
You should find a place as lovely as possible, sheltered from the wind and close to the sea, and as isolated as possible from other people to enjoy a little privacy. Optimizing even more, it could be a good thing to leave no good spots too close to you because a new family coming to the beach could have the incentive to move there and ruin your  spot.

\medskip
Sit on a bus, a train, or an auditorium: Where should you sit to enjoy privacy and comfort so that others arriving later will have no incentive to sit too close?
Another (more technical and less antisocial) analogy is to imagine a load-balancing context with local influences. Assume that there are $n$ processors with $m$ computation tasks to solve but that solving a task on a processor impacts other processors around (say, for instance, that nearby processors share a resource of energy). How should the tasks be distributed to processors? Can this distribution be efficiently computed?

\subsection{Our model}

Many problems in wireless communication networks have been modelled using game-theoretic approaches and
 have received considerable attention (see \cite{lasaulce2011game} for an overview). To handle access to the communication medium, one must consider interference problems. For example, in the base station selection \cite{aryafar2013rat}, a user's cost (throughput) depends simultaneously on the number of users associated with the base station and its neighbourhood.

\medskip
To model these situations, we introduce the notion of {\it Neighbourhood Balancing Game} (NBG).
In such a game, players must choose a single resource in a set of $n$ possibilities. They endure a cost depending on how many players chose the same resource, but also nearby resources (as opposed, for instance, to congestion games where the cost of players depends only on the number of players who chose the same resource).
Thus stated, the model of NBG is quite general, so we shall consider particular cases and use graph-theoretic terminology to model this notion of proximity between resources and infer results from the graph topology.

We shall be very interested in {\it equilibria}, configurations of players' choices such that no player can improve its situation by changing location if other players do not move. We shall examine these equilibria's existence, structure, and computational tractability.

While we can develop a similar theory with atomic players, here we focus on the {\it nonatomic case} and consider a continuum of players, or {\it mass distribution}, to avoid side effects based on divisibility issues. As will be seen, this case is rich enough to provide many interesting examples.

Our model of nonatomic Neighbourhood Balancing Games lies somewhere between the model of nonatomic selfish routing (or "Wardrop model", see \cite{wardrop1952road}), selfish load balancing and congestion \linebreak games.

As will be seen, our framework is more straightforward than the previous ones in many cases, but many examples already manifest a great depth of complexity.

\subsection{State of the art}

The general framework of neighbourhood balancing games extends the congestion games defined by Rosenthal~\cite{Rosenthal1973class}. Congestion games are non-cooperative games in which the players compete to share a set of resources, and the cost of each player depends on the number of players choosing the same resource. There are two types of games: atomic (a finite number of players) and nonatomic (a continuous number of players).

%\medskip
For nonatomic congestion games, a well-known game theoretic traffic model is due to
Wardrop~\cite{wardrop1952road}. Wardrop equilibria have been introduced to model network behaviours in which travellers (for transportation context) or packets (telecommunication
context) choose routes they perceive as the shortest (see survey~\cite{correa2011wardrop} for more details).
This model was conceived to represent road traffic with the idea of an infinite number of agents responsible for an infinitesimal amount of traffic each.
Beckmann {\it et al.} \cite{beckmann1956studies} proved that Wardrop equilibrium is a solution of a convex optimization program and is unique.

In algorithmic game theory, a question arises about the impact of the degradation of the social cost of agents, taking into account their interests. This measure is called the {\it price of anarchy}~\cite{koutsoupias1999worst} (ratio between the worst equilibria and the
optimal solution).

This question has been intensively studied for congestion games.

%%\medskip
For nonatomic congestion games with affine costs, the price of anarchy is upper-bounded by $4/3$, and this bound is sharp \cite{roughgarden2002bad}. Indeed, the simple game introduced by Pigou is defined as a network with two parallel routes composed of a single edge connecting a source to a destination. Its price of anarchy is $4/3$. Moreover, using the same topology and considering that the network's cost functions are polynomials of degree at most $p$, Roughgarden \cite{roughgarden2002price} proved that
 the price of anarchy for such networks is large (the order of magnitude is $\Theta(p/\log p)$).
 Recently, some extensions of the price of anarchy have been thought about in this type of network, considering the traffic variation~\cite{colini2016price,colini2020selfish,wu2021selfishness}.

For atomic congestion games, the games that seem most similar to those studied in this work are
 load balancing games, or {\it selfish load balancing}. The underlying model, known as the {\it KP-model}, was introduced by Koutsoupias and Papadimitriou in \cite{koutsoupias1999worst}, where they define measures for the quality of equilibria. These problems have been studied widely (see chapter 20 of \cite{nisan2007algorithmic} for a survey).
 Indeed, besides their conceptual simplicity (some machines are shared between selfish users who decide which machines they will assign their tasks to), these problems are crucial in distributed environments. Koutsoupias and Papadimitriou proved the existence of Nash equilibria and computed the price of anarchy when machines are identical. Czumaj and V\"{o}cking \cite{CzumajVocking:Tight-bounds-for-worst-case} gave the tight bound $\Theta(\log m/\log \log m)$ for the price of anarchy on $m$ uniformly related machines.
In such environments, Christodoulou \emph{et al.} \cite{ChristodoulouKoutsoupiasNanavati:Coordination-Mechanisms} introduced the coordination mechanisms (a set of scheduling policies, one
for each machine) to obtain socially desirable solutions despite the selfishness of the agents. Several works analyzed the existence of pure Nash equilibria (see \cite{durr2009non,ImmorlicaLiMirrokni:Coordination-Mechanisms-for-Selfish,Thang:NP-hardness-of-pure-Nash} for example) and their prices of anarchy~\cite{abed2014optimal,AwerbuchAzarRichter:Tradeoffs-in-worst-case-equilibria, GairingLuckingMavronicolas:Computing-Nash-equilibria, roughgarden2015intrinsic,SchuurmanVredeveld:Performance-guarantees-of-local} for these models (for uniform machines or unrelated machines, for weighted or not players).

Moreover, the \textit{price of stability}~\cite{schulz2003performance} (ratio between the best equilibria and the optimal solution) is another measure when the game admits several equilibria.  Since nonatomic congestion games have a unique equilibrium, their prices of anarchy and stability are identical. This reason is why it has been studied only for atomic congestion games; the price of stability is known for congestion games with linear cost functions \cite{caragiannis2006tight, christodoulou2005price} and is upper-bounded \cite{christodoulou2015price} for  congestion games with polynomial cost
functions.

One of the closest games related to our model is the facility location game introduced by Vetta~\cite{vetta2002nash}.
Competitive facility location games deal with the placement of sites by competing market players. The \emph{facility location game} plays on weighted bipartite graphs in which each player chooses to open a single facility within the set of facilities the suppliers will serve, according to the distribution of users on some vertices. Given a strategy profile, a supplier $s$ serves the facilities closest to $s$ and receives a payment from these facilities. Vetta~\cite{vetta2002nash} proved that the facility location game always admitted a Nash equilibrium and gave an upper bound on the price of anarchy. This has been extended in literature in several papers.

\medskip
The \emph{discrete Voronoi game} corresponding to a simple model for the competitive facility location is very similar to neighbourhood balancing games. The discrete Voronoi game plays on a given graph with two players. Every player has to choose a set of vertices, and every vertex is assigned to the closest player. D\"urr and Thang~\cite{durr2007nash}, and Teramoto \emph{et al.}~\cite{teramoto2006voronoi} independently proved that deciding the existence of a Nash equilibrium for a given graph is NP-hard. Several works have been devoted to extending similar results on various types of graphs (cycles~\cite{durr2007nash}, trees~\cite{bandyapadhyay2015voronoi}).

\subsection{Outline of the paper}
Section~\ref{sec:model} is devoted to defining the games that we shall call NBG and the variants we shall study. Section~\ref{sec:equi} will describe the notion of equilibrium and $\delta$-strong equilibrium. We shall also prove a sufficient condition on the costs for the game to admit an equilibrium. Moreover, in Subsection~\ref{sec:complexite}, we prove that knowing whether a game admits a strong equilibrium is NP-complete. Furthermore, we adapt the notion of potential to our games to prove that symmetric graphical games admit a potential function. Subsections~\ref{sec:qualite:1} and~\ref{sec:qualite:2} focus on the efficiency of equilibria. Finally, in Section~\ref{sec:graph}, we focus on characterizing equilibria in games where very simple graphs represent interactions between resources.

\section{Neighbourhood Balancing Games (NBGs)} \label{sec:model}

In this section, we present and discuss our model, starting from the general case and refining particular cases and properties that will be needed later.

\subsection{General model}

Let $n \geq 1$ be an integer. Integers from $1$ to $n$, whose set we denote as $[n]$, will be called {\it vertices}$\,$; those are the resources (we use a graph-theoretical vocabulary for reasons that will appear soon). In the paper, symbols $i,j$ are for vertices; if omitted, their scope is the set $[n]$.

\medskip
A {\it mass distribution} on $[n]$ is a vector $\x=(x_1, x_2, \cdots, x_n)$ with real nonnegative entries, where $x_i$ is to be thought of as the mass on vertex $i$. The {\it total mass} of such a distribution is the sum
of all $x_i$'s for ${i \in [n]}$ and will be denoted by $r$.
A mass distribution can be considered as a continuum of {\it players}, each choosing a single vertex as a {\it strategy}. We also refer to such an imaginary player by the expression "infinitesimal unit of mass".

We shall denote $\Delta_r(n)$ the set of $n$-dimensional mass distributions
with total mass $r>0$. Given $\x \in \Delta_r(n)$, we say that a vertex $i$ is {\it charged} if $x_i >0$, otherwise it is {\it uncharged}. The {\it support} of $\x$ is the set of charged vertices.

\medskip
\noindent A {\it cost function} on vertex $i$ is a function
\begin{equation*} C_i : \x \mapsto C_i(\x) \in \R^+ \end{equation*}
mapping every mass distribution $\x \in \Delta_r(n)$ (for a given $r$) to a nonnegative real number, the {\it cost of
 vertex $i$ under mass distribution $\x$}.

\medskip
We interpret the cost $C_i(\x)$ as the price paid by individuals in the continuum of players that have chosen vertex $i$ as their strategy. Of course, each player has a goal of minimizing such a cost. This can also be viewed as people in location $i$ enduring a nuisance, based on the total number of people in their location and other locations.

\medskip
To aggregate all this, we define a {\it Neighbourhood Balancing Game} (NBG for short) as a triple $(n, r, C)$, where $n \geq 1$, $r >0$, and $C = (C_i)_{i \in [n]}$ is
the $n$-dimensional vector of cost functions $C_i : \Delta_r(n) \rightarrow \R^+$.

Given this, the definition is quite general, and we shall be interested in NBGs where cost functions have prescribed forms, which we shall describe afterwards.

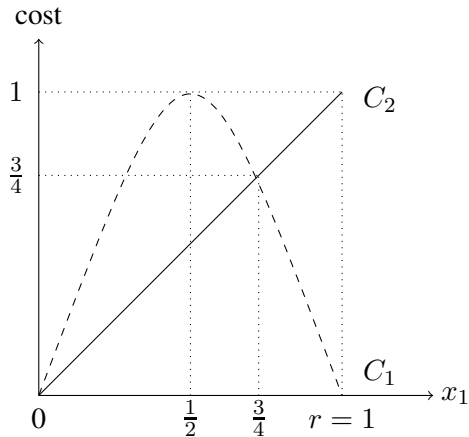
\begin{figure}[!b]
\vspace*{-2mm}
\centering
\begin{tikzpicture}
\node (origine) at (0,-0.3) {0};
\node at (4,-0.3) {$r=1$};
\node at (-0.3,4) {$1$};
\node at (2.9,-0.3) {$\frac{3}4$};
\node at (2,-0.3) {$\frac{1}2$};
\node at (-0.3,2.9) {$\frac{3}4$};
\node (x_1) at (5.5,0) {$x_1$} ;
\node (cost) at (0,5.0) {cost} ;
\draw[dotted] (4,0) -- (4,4) ;
\draw[dotted] (2,0) -- (2,4) ;
\draw[dotted] (0,4) -- (4,4) ;
\draw[dotted] (0,2.9) -- (2.9,2.9) ;
\draw[dotted] (2.9,0) -- (2.9,2.9) ;
\draw[->] (0,0) -- (5.2,0) ;
\draw[->] (0,0) -- (0,4.7) ;
\node at (4.5,3.9) {$C_2$} ;
\node at (4.5,0.3) {$C_1$} ;

\draw (0,0) -- (4,4) ;
\draw[dashed] (0,0) .. controls (2,5.3) .. (4,0) ;
%\node at (5.5,1) {$C_1$} ;
%\node at (5.5,2) {$C_2$} ;
\end{tikzpicture}
\caption{Costs functions for the example in Subsection \ref{sub:firstex}. The cost $C_1$ (resp. $C_2$) is represented by the dotted curve (resp. solid).  The $x$-axis is the mass of vertex $1$. }
\label{fig:ex1}
\end{figure}

\subsubsection{A first example: two-commodity dilemma game (see Figure~\ref{fig:ex1}) } \label{sub:firstex}

Let us consider a first, simple NBG: the NBG $(2, 1, C)$ with $C_1(x_1,x_2) = 4 x_1 x _2$ and $C_2(x_1,x_2)=x_1$. This can be thought of as a two-commodity dilemma with a total mass $1$ of players simultaneously deciding between commodities $1$ and $2$, then enduring the cost $C_i$ depending on their choice and the repartition of the mass.

\noindent Since for all $(x_1,x_2) \in \Delta_1(2)$ we have $x_1+x_2=1$, we can rewrite these costs as functions of $x_1$ only i.e. $C_1(x_1)= 4x_1(1-x_1)$ ans $C_2(x_1)=x_1$; these are depicted in Figure~\ref{fig:ex1}. As can be seen, costs are equal when $x_1=0$ and $x_1=\frac{3}{4}$. When $x_1 < \frac{3}{4}$, vertex $2$ has a lower, hence better cost. Hence, the mass on vertex $1$ might want to move to $2$, decreasing $x_1$ in the process. When $x_1 > \frac{3}{4}$, the cost is better on vertex $1$, hence the mass on $2$ might want to move to $1$, increasing $x_1$. Thus, we can expect that if we let the mass continuously move, starting from $x_1 \neq \frac{3}{4}$, we would obtain stabilization at either $x_1=0$ or $x_1=1$ depending on the case. If $x_1$ is exactly $\frac{3}{4}$, the two costs are equal, and we can interpret this as a case where no player can improve its cost by switching from one vertex to the other.

\medskip
We shall define mass distributions such as $x_1=0$, $x_1=\frac{3}{4}$ and $x_1=1$, where no mass has the incentive to move, as equilibria, and it will be our main endeavour to try and understand them; see Section~\ref{sec:equi} for a definition.

\subsection{Graphical NBGs}

\subsubsection{Definition}

In this section, we define a subclass of NBGs, which will be our main focus for the rest of the paper.
In addition to being a natural and more straightforward case, Subsection~\ref{sub:whygraphical} will  explain the reason for this particular class.

\medskip
We now consider only the cost functions of the form
\[C_i(\x) = f_i(x_i) + \sum_{j \neq i} \alpha_{j,i} x_j\]
where:
\begin{enumerate}
\renewcommand\labelenumi{\theenumi}
     \renewcommand{\theenumi}{(\roman{enumi})}
%%%%% R^+
\item $f_i : [0,r] \rightarrow \R^+$ is a continuous non-decreasing function, such that $$t > 0 \Rightarrow f_i(t) > 0$$ (hence values
 are positive except possibly for $t=0$);
\item $\alpha_{j,i}$, for every $i \neq j$ in $[n]$, are nonnegative real
 numbers.
\end{enumerate}
Functions $f_i$ are called \emph{vertex-cost functions} and are not to be confused
with proper cost functions $C_i$ that are deduced from $f_i$'s and $\alpha_{i,j}$'s.

\medskip
We call such a game a \emph{graphical} NBG and it will be denoted \[(n, r , f,\alpha),\]
where $f = (f_i)_{i \in [n]}$ and $\alpha = (\alpha_{i,j})_{i\neq j \in [n]}$.
If furthermore for all $i,j \in [n]$, $i\neq j$, we have $\alpha_{j,i} = \alpha_{i,j}$, we say that the game is \emph{symmetric}.

\medskip
To a graphical NBG, we can associate an \emph{underlying graph} (hence the name), which is either
a directed graph $(V, A)$ (in the general, non-symmetric case) or an undirected graph
$(V,E)$ (if the game is symmetric), where $V=[n]$ and
either
\[A = \{(i,j) : 1 \leq i \neq j \leq n \text{ and } \alpha_{i,j} > 0\} \]
is the set of arcs in the directed case or
\[E = \{ \{i,j\} : 1 \leq i \neq j \leq n \text{ and } \alpha_{i,j} > 0\} \]
is the set of edges in the symmetric case.

\medskip
Let us go back to the example in Subsection~\ref{sub:firstex}. It is not graphical since $C_1(\x) = 4x_1 x_2$ is not of the prescribed form. If we had $C_1(\x)= x_1 +x_2$ and still $C_2(\x)=x_1$, the example would have been graphical and even symmetric, with
$f_1(x_1) = x_2$, $f_2(x_2) = 0$ and $\alpha_{1,2} = \alpha_{2,1} = 1$.

\subsubsection{Special case where there are two resources ($n=2$)}

The case $n=2$ is, of course, more straightforward than the general case but is still rich enough
to provide examples and counterexamples.
We shall represent a symmetric graphical NBG as:

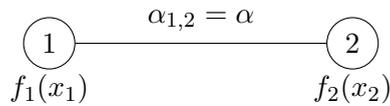
\begin{center}
 \begin{tikzpicture}
  \node[draw,circle] (a) at (0,0) {$1$};
  \node[draw,circle] (b) at (4,0) {$2$} ;
  \draw (a) edge (b) ;
  \node (t) at (2,.3) {$\alpha_{1,2}= \alpha$} ;
  \node (f1) at (0,-.6) {$f_{1}(x_1)$} ;
  \node (f2) at (4,-.6) {$f_2(x_2) $} ;
 \end{tikzpicture}
\end{center}

The following proposition will be used to build examples. It shows that in the case
$n=2$, if some regularity conditions are satisfied,
we can always find $\alpha_{1,2}$ and vertex-cost functions to obtain
any couple of cost functions on vertices, up to translation.

\begin{proposition} \label{prop:exbuilder}
Let $D_1$ and $D_2$ be two continuously differentiable functions on $\Delta_r(2)$ for a given $r>0$. We can find a real number $c$, a positive $\alpha>0$ and two functions $f_1$ and $f_2$
such that $(2,r,f=(f_1,f_2),\alpha)$ is a symmetric graphical NBG with cost functions $C_1 = D_1 + c$ and $C_2 = D_2 + c$.
\end{proposition}

\begin{proof}
For $x \in \Delta_r(2)$ we want
\[ D_1(x_1,x_2) + c = C_1(x_1,x_2) = f_1(x_1) + \alpha x_2 = f_1(x_1) + \alpha (r-x_1),\]
so we define on $[0,r]$
\[ f_1(x_1) = D_1(x_1,r-x_1) + c - \alpha (r - x_1) \]
and likewise
\[ f_2(x_2) = D_2(r-x_2,x_2) + c - \alpha (r - x_2) \]
where $c$ and $\alpha$ are to be defined.

What remains to do is to find values $\alpha \geq 0$ and $c$ such that $f_1$ and $f_2$ will be positive
and non-decreasing. Both functions are continuously differentiable and
for $x_1 \in [0,r]$
we have
\[f'(x_1) = \frac{ \partial D_1}{\partial x_1}(x_1,r-x_1) - \frac{ \partial D_1}{\partial x_2}(x_1,r-x_1) + \alpha, \]
hence if we want $f_1$ and $f_2$ non-decreasing it is enough to use any $\alpha$ satisfying
\begin{align*}
\alpha \geq \max \Big( & \max_{x \in [0,r]} \left\{ -\frac{ \partial D_1}{\partial x_1}(x,r-x) + \frac{ \partial D_1}{\partial x_2}(x,r-x) \right\}, \\[5pt]
  & \max_{x \in [0,r]} \left\{ \frac{ \partial D_2}{\partial x_1}(r-x,x) - \frac{ \partial D_2}{\partial x_2}(r-x,x)\right \} \Big)
\end{align*}
and then we can easily find $c$ such that both functions are positive.
\end{proof}

\subsection{Affine, linear, normal, uniform NBGs}
We shall restrict graphical NBGs by considering the following definitions for a graphical NBG $G = (n,r,f,\alpha)$:
\begin{itemize}
\item $G$ is \emph{affine} if all vertex-cost functions $f_i$ for $i \in V$ are of
the form
\[ f_i(x_i) = \alpha_{i,i} x_i + b_i,\]
where $\alpha_{i,i}, b_i$ are nonnegative. Cost functions are then of the form
\[C_i(\x) = b_i + \sum_{j \in V} \alpha_{j,i} x_j\]
hence are also affine functions. An affine graphical NBG can be symmetric or not.
\item if $G$ is affine and moreover all $b_i$'s satisfy $b_i=0$, then
the NBG is \emph{linear} and all cost
functions $C_i$ for $i \in V$ are of
the form
\[C_i(\x) = \sum_{j \in V} \alpha_{j,i} x_j ,\]
where the $\alpha_{i,j}$'s are nonnegative. A linear NBG is said to be \emph{normal} if  $\alpha_{i,i} = 1$ for all $i$.

\item $G$ is an $\alpha$-{\it uniform} NBG if it is a normal linear NBG where all nonzero
$\alpha_{i,j}$'s for $i \neq j$ have a common value $\alpha$, or in other words all cost functions
$C_i, i\in V$, are of the form

\[C_i(\x) = x_i + \alpha \sum_{j \in N[i]} x_j,\]
where $N[i] \subseteq V\setminus\{i\}$ is a set of vertices which we call
{\it neighbours}
of $i$ using graph-theoretic terminology. If the game is symmetric,
these neighbours are precisely the neighbours in the {\it underlying undirected} graph (otherwise, they are the in-neighbours in the underlying directed graph).
\end{itemize}

\smallskip
We shall recall these definitions when needed.

\subsection{Summary}

Here is a table to sum up the different cases:

\begin{center}
\scalebox{0.95}{
\begin{tabular}{|c|c|}
\hline
Class of NBG   & condition on cost functions $C_i$ \\
\hline
\hline
general & $C_i(\x) \geq 0$ \\
\hline
     & $C_i(\x) = f_i(x_i) + \sum_{j \neq i} \alpha_{j,i} x_j$ \\
graphical & $f_i$ continuous non-decreasing \\
     & $t > 0 \Rightarrow f_i(t) > 0$ \\
     & $f_i(0)\geq 0$, $\alpha_{j,i} \geq 0$ \\
\hline
affine  & $C_i(\x) = b_i + \sum_{j } \alpha_{j,i} x_j$ \\
     & $b_i\geq 0$, $\alpha_{j,i} \geq 0$ \\
\hline
linear & $C_i(\x) = \sum_{j } \alpha_{j,i} x_j$ \\
    & $\alpha_{j,i} \geq 0$ \\
\hline
normal & linear with $\alpha_{i,i} = 1$ \\
\hline
$\alpha$-uniform & normal linear with \\
         & $j\neq i, \alpha_{j,i} > 0 \Rightarrow \alpha_{j,i} = \alpha$ \\
\hline
\end{tabular} }
\end{center}

Each class contains the classes below it in the table, and all classes included in the Graphical class can be either symmetric ($\alpha_{j,i} = \alpha_{i,j}$ for all $i \neq j$) or not.

\section{Equilibria} \label{sec:equi}

Equilibria of NBGs are our main focus in this paper. An {\it equilibrium} is a mass distribution such that no infinitesimal player has the incentive to move because, in that distribution, every player has made the best choice for itself. This notion relates to the $n$-player games definition of Nash equilibria and the definitions of equilibria in other nonatomic games such as Wardrop games or nonatomic congestion games. We consider only {\it pure} equilibria here because infinitesimal players do not use mixed strategies.

%%\medskip
In this section, after properly defining equilibria and proving their existence in the general case, we establish an algorithmic complexity result for the existence of a refined notion of equilibria,   $\delta$-{\it strong equilibria}. We then introduce the notion of the potential function and prove that symmetric graphical games enjoy the existence of such a function, which is a tool to compute equilibria.

\subsection{Definition}

An {\it equilibrium} is a mass distribution such that no infinitesimal mass quantity can get a lower cost by moving from one vertex to another. Formally,
$\x^* \in \Delta_r(n)$ is an equilibrium if
\[ \forall i,j \in [n], \quad x^*_i > 0 \Rightarrow C_i(\x^*) \leq C_j(\x^*).\]
This definition implies that in an equilibrium, all charged vertices share
the same cost, while uncharged vertices have at least the same cost as the charged vertices.

\eject

For $\delta \geq 0$, a {\it $\delta$-strong equilibrium} is a $\x^* \in \Delta_r(n)$ such that no mass quantity $0 \leq \epsilon \leq \delta$ can improve its cost by
moving from one vertex to another. Formally,
\begin{equation} \label{ALeq1}
 \forall 0 \leq \epsilon \leq \delta, \forall i,j \in [n] , x^*_i
 \geq \epsilon \Rightarrow C_i(\x^*) \leq C_j(\x^* - \epsilon \cdot {\mathbf e}_i +
 \epsilon \cdot {\mathbf e}_j),
 \end{equation}
where $({\mathbf e}_i)_{i \in [n]}$ is the canonical basis of $\R^n$. If $\delta \leq \delta'$, a $\delta'$-strong equilibrium is $\delta$-strong, particularly all $\delta$-strong equilibria are $0$-strong, which amounts to saying that it is an equilibrium in the sense defined above.

We say that an equilibrium is {\it strong} if it is $\delta$-strong for some $\delta >0$. See Figures~\ref{fig:ex_eq1} and \ref{fig:ex_eq2} for examples of equilibria,
strong or not. Thanks to Proposition~\ref{prop:exbuilder}, we do not need to define these examples explicitly where only the curves' relative positions and shapes are relevant.

\begin{figure}[!h]
\centering
\scalebox{0.95}{
\begin{tikzpicture}

\node (origine) at (0,-0.3) {0};
\node at (5,-0.3) {$r$};
\node at (3.9,-0.3) {$\gamma$};
\node (x_1) at (6.5,0) {$x_1$} ;
\node (cost) at (0,5.0) {cost} ;

\draw[dotted] (3.92,0) -- (3.92,2.4) ;
\draw[dotted] (5.0,0) -- (5.0,2.4) ;
\draw[->] (0,0) -- (6.2,0) ;
\draw[->] (0,0) -- (0,4.7) ;
\draw (0,4) -- (5,2) ;
\draw[dashed] (0,4) .. controls (2,5) .. (5,1) ;
\node at (5.5,1) {$C_1$} ;
\node at (5.5,2) {$C_2$} ;
\end{tikzpicture} }\vspace*{-2mm}
\caption{We consider a game with $2$ vertices. The curves represent the costs of the two vertices as a function of the mass on vertex $1$.  There are three equilibria in this case: $x_1 = 0$, $x_1 = \gamma$ (second intersection of the costs curves) and $x_1=r$. Equilibria $x_1 = 0$ and $x_1 = r$ are strong; however $x_1 = \gamma$ is not strong since a small quantity $\epsilon$ of mass can always move from $2$ to $1$ and improve its cost.}
\label{fig:ex_eq1}
%%\end{figure}

%%\begin{figure}[h]
\vspace*{4mm}
\centering
\scalebox{0.95}{\begin{tikzpicture}
\node (origine) at (0,-0.3) {0};
\node at (5,-0.3) {$r$};
\node (x_1) at (6.5,0) {$x_1$} ;
\node (cost) at (0,4.7) {cost} ;

\draw[->] (0,0) -- (6.2,0) ;
\draw[->] (0,0) -- (0,4.4) ;
\draw[dotted] (5.0,0) -- (5.0,2.4) ;
\draw (0,4) -- (5,1) ;
\draw[dashed] (0,4) to (5,2) ;
\node at (5.5,2) {$C_1$} ;
\node at (5.5,1) {$C_2$} ;
\end{tikzpicture} }\vspace*{-2mm}
\caption{We consider a game with $2$ vertices. The curves represent the costs of the two vertices as a function of the mass on vertex $1$. There is only one equilibrium in $x_1 = 0$, which is not strong. }
\label{fig:ex_eq2}
\end{figure}

\subsection{Existence of equilibria} \label{sec:genEqExistence}

Here, we consider general cost functions
and show the existence of an equilibrium under the condition of continuity. The proof is adapted from Nash's proof of the existence of
a mixed symmetric Nash equilibrium in a symmetric game~\cite{nash1950equilibrium}.

\begin{theorem} \label{th:existence:g}
 Let $(n,r,C)$ be a general NBG, and suppose that cost functions
 $(C_i)_{i \in [n]}$ are continuous. Then, the game admits an equilibrium.
\end{theorem}

\begin{proof}
 Define, for all $i,j \in [n]$, a function $g_{i,j} : \Delta_r(n) \rightarrow \R^+$ by
 \[g_{i,j} (\x) = x_j \cdot \max\left(0, C_j(\x) - C_i(\x) \right)\]
 and $F_i : \Delta_r(n) \rightarrow \R^+ $
 by
 \[ F_i(\x) = \frac{x_i + r \cdot \sum_{j \in [n]} g_{i,j}(\x)}{1 + \sum_{i,j \in [n]} g_{i,j}(\x)}.\]
 This definition guarantees that $F = (F_1, F_2, \cdots, F_n)$ takes values in $\Delta_r(n)$ and is continuous,
 hence
 has a fixed point $\x^*$ by Brouwer's fixed point theorem. Consider $k \in [n]$ such
 that $C_k(\x^*) = \max\{C_j(\x^*) : x^*_j > 0 \}$ and $x^*_k>0$. Then
 \[ C_j(\x^*) > C_k(\x^*) \Rightarrow x^*_j = 0,\]
 so that $g_{k,j} (\x^*) = 0$ for all $j\in [n]$. In particular, we have
 $F_k(\x^*) = x^*_k$, hence
 \[ \frac{x^*_k}{1 + \sum_{i,j \in [n]} g_{i,j}(\x^*)} = x^*_k.\]
 Since $x^*_k>0$, this implies that $g_{i,j}(\x^*) = 0$ for all $i,j \in [n]$.
 In particular, for $i,j$ with $x^*_j > 0$, we have $\max(0,C_j(\x^*) - C_i(\x^*)) = 0$,
 which is the definition of an equilibrium.
\end{proof}

\subsection{Games with no equilibria}

A non-continuous game does not need to have an equilibrium. Consider for instance $(2,1,C)$ with $C_1(\x)=1$, and $C_2(\x) = 2$ when $x_1 < \frac{1}{2}$, and $C_2(\x) = 0$ otherwise.
When $x_1 < \frac{1}{2}$, we have $x_2 >0$ but $C_1(\x) < C_2(\x)$, hence not an equilibrium, and for $x_1 \geq \frac{1}{2}$ we have $x_1 >0$ and $C_2(\x) < C_1(\x)$ hence also not an equilibrium.

\subsection{Structure and complexity of equilibria} \label{sec:complexite}

Here, we study the structure of equilibria in one particular case, namely normal linear NBGs. This enables us to derive complexity results for the problem of finding an equilibrium in a given NBG.

% \subsubsection{Normal linear NBGs}
So we consider a (possibly non-symmetric) normal linear NBG $G$ and its underlying directed graph $D$, i.e. its cost functions are of the form
\begin{equation} \label{eq:normalNBG}
C_i(\x) = x_i + \sum_{j \neq i} \alpha_{j,i} x_j, \vspace*{-2mm}
\end{equation}
where $\alpha_{i,j}$ may be different from $\alpha_{j,i}$.

\medskip
Let us recall that for a directed graph $D=(V,A)$, a {\it kernel} is a set of
vertices
$K \subseteq V$ such that:

\begin{itemize}
\item for any two distinct vertices $v,w \in K$, arcs $(v,w)$ and $(w,v)$ do not
 belong to $A$ (hence $K$ is a directed {\it stable} set);
\item if $z \not\in K$, then there is $v \in K$ with $(v,z) \in A$ ($K$ is a
 directed {\it dominating} set).
\end{itemize}

\begin{proposition}
 Let $G=(n,r,C)$ be an NBG with normal linear cost functions and underlying graph $D$; suppose furthermore that
 \[ \; \alpha_{i,j} > 0 \Rightarrow \alpha_{i,j} > 1 \; \;and \; \;\alpha_{i,j} + \alpha_{j,i} >2 \;\; (i\neq j). \]
 Then, the supports of strong equilibria are exactly the kernels of $D$.
\end{proposition}

\begin{proof}
 First, consider a kernel $K$ of $D$ of size $k$ and define $\x^*$
 by $x_i^* = \frac{r}{k}$ if $i \in K$ and $x_i^* = 0$ otherwise. Since $K$ is a stable set, the cost inside $K$ is exactly $\frac{r}{k}$, and since $K$ is dominating and any nonzero $\alpha_{i,j}$ is at least one, the cost is at least $\frac{r}{k}$ outside $K$. Let
 $0 \leq \epsilon \leq \frac{r}{k}$ and consider a mass $\epsilon$ moving from
 vertex $i \in K$ to $j$.

\medskip
 If $j \in K$, the cost in $j$ after the change will be $\frac{r}{k} + \epsilon$, which is worse. If $j \not\in K$, then there is an arc $(i',j)$ from a vertex $i' \in K$ to $j$. The worst case is when $i'=i$, and in this case, the cost in $j$ will be at least
 \[ \epsilon + \alpha_{i,j} (\frac{r}{k}-\epsilon) = (\alpha_{i,j}-1)(\frac{r}{k} - \epsilon) + \frac{r}{k} \geq \frac{r}{k}, \]
 which is also not better.
 Hence, $K$ is the support of a $\frac{r}{k}$-strong equilibrium.

\medskip
 Conversely, suppose now that $\x^*$ is a $\delta$-strong equilibrium of $G$ for a $\delta >0$, and $i,j \in [n]$ are two charged
 vertices with $\alpha_{i,j} >
 0$, i.e. $(i,j)$ is an arc of $D$. Consider an $\epsilon > 0$ as
 in Definition (\ref{ALeq1}), and suppose furthermore that $\epsilon < \min(\delta, x_i^*,x_j^*)$.
 Then
 \[C_i (\x^*) \leq C_j(\x^* + \epsilon \cdot {\mathbf e_j} - \epsilon \cdot {\mathbf e_i})\]
 and
 \[C_j (\x^*) \leq C_i(\x^* + \epsilon \cdot {\mathbf e_i} - \epsilon \cdot {\mathbf e_j}).\]
 Let $A_i = \sum_{k \neq i,j} \alpha_{k,i} x_k^*$ and
 $A_j = \sum_{k \neq i,j} \alpha_{k,j} x_k^*$. Then, the inequalities above can be
 written
 \[x_i^* + \alpha_{j,i} x_j^* + A_i \leq x_j^* + \epsilon + \alpha_{i,j} (x_i^* -
  \epsilon) + A_j\]
 and
  \[x_j^* + \alpha_{i,j} x_i^* + A_j \leq x_i^* + \epsilon + \alpha_{j,i} (x_j^* - \epsilon) + A_i.\]
  \eject

 \noindent Summing these two inequalities, we obtain after cancellations
 \[0 \leq 2\epsilon - \epsilon (\alpha_{i,j} + \alpha_{j,i}), \]
 so that $\alpha_{i,j} + \alpha_{j,i} \leq 2$, contradicting our hypothesis.
 So, there can be no arc between charged vertices,
 and the support of the equilibrium is a stable set.

\medskip
 Moreover, if $j$ is uncharged, then its cost must be positive by the definition
 of  equilibria; hence, there must be some charged vertex $i$ with
 $\alpha_{i,j} > 0$. Therefore, the support of a strong equilibrium is a kernel of $D$. \end{proof}

\noindent For a normal linear NBG or even an affine NBG, checking if a mass distribution is a strong equilibrium
is a simple task and amounts to checking a quadratic number of inequalities.
Now, observe that determining whether a given digraph admits
kernels is NP-complete (\cite{chva73}, \cite[p. 204]{gare79}). Also, there is a straightforward polynomial reduction from kernels
to strong equilibria by changing
a digraph into a normal linear NBG simply by choosing some $\alpha_{i,j} > 1$ on all arcs. Hence:

\begin{corollary}
 Determining whether an affine NBG admits a strong equilibrium is NP-complete.
\end{corollary}
%% $C_j(x_\epsilon) + \alpha_{ij}(x_i -\epsilon) \leq C_i(x_\epsilon) + \alpha_{ij}x_i$
In general, checking if a distribution is a strong equilibrium depends on the model used to define functions. However, we still have that:

\begin{corollary}
 Determining whether an NBG admits a strong equilibrium is NP-hard.
\end{corollary}

As a conclusion to this subsection, let us mention a simple example where the game admits an equilibrium (not strong), and the graph has no kernel. We consider the normal linear NBG on three vertices \{1,2,3\} with arcs $(1,2)$, $(2,3)$ and $(3,1)$, with a coefficient $\alpha >0$ on each arc. Consider the mass distribution $(1/3,1/3,1/3)$: it is an equilibrium. However, this equilibrium is not strong, and the graph admits no kernel.

\subsection{Potential function}

Let $G = (n,r,C)$ be an NBG. A {\it potential function} for $G$ is a differentiable
function $\Phi$ defined on a neighbourhood of $\Delta_r(n)$ such that
\[\forall i \in [n], \frac{\partial \Phi(\x)}{\partial x_i} = C_i(\x).\]
A potential function is a tool to study equilibria. In particular

\begin{proposition} \label{prop:minPhi}
 If an NBG, $G = (n,r,C)$, admits a potential function $\Phi$, then the local minima of $\Phi$ on $\Delta_r(n)$ are equilibria of the game.
\end{proposition}

\begin{proof}
 Let $\x^*$ denote a local minimum of $\Phi$ on $\Delta_r(n)$. In
 particular, if $x_i^* > 0$ and $j \in [n]$, $j \neq i $, there
 is an $\epsilon >0$ such that if $x_i^* - \epsilon \geq 0$, we have
 \[ \Phi(\x^*) \leq \Phi(\x^* - t \cdot \epsilon \cdot {\mathbf e_i} +
  t \cdot \epsilon \cdot {\mathbf e_j}) \]
 for all $t \in [0,1]$. Defining
 \[ f(t) = \Phi(\x^*) - \Phi(\x^* - t \cdot \epsilon \cdot {\mathbf e_i} +
  t \cdot \epsilon \cdot {\mathbf e_j}), \]
 we see that $f$ is differentiable on $[0,1]$ and is nonpositive while $f(0) =
 0$. Hence, we have $f'(0) \leq 0$ and thus
 \[ \epsilon \cdot C_i(\x^*) - \epsilon \cdot C_j(\x^*) \leq 0,\]
 which is the definition of equilibria.
\end{proof}

\noindent {\bf Remark.} Following the proof above, an interpretation of the potential function $\Phi$ is that for $|t|$ small enough, if $x_i >0$ we have
\[\Phi(\x+t\cdot {\mathbf e_i} - t\cdot {\mathbf e_j}) - \Phi(\x) = t \cdot (C_i(\x) - C_j(\x)) + o(|t|) ,\]
hence the total quantity of cost variation when moving between $i$ and $j$
can be deduced from the difference of potential.

\medskip
\noindent {\bf Remark.} An equilibrium corresponding to the minimum of a potential function is not necessarily strong. See Figure~\ref{fig:ex_eq2} for an example with a unique, not strong, equilibrium,
which must be the global minimum of the potential function (this example is a graphical symmetric NBG, so it admits a potential function by Proposition~\ref{prop:potentialFormula}).

\medskip
\noindent {\bf Remark.} The converse is false: an equilibrium does not always
correspond to a local minimum of the potential. See a counterexample in Figure~\ref{fig:eq_vs_phi}.

\begin{figure}[!h]
\centering
\begin{tikzpicture}

\node (origine) at (0,-0.3) {0};
\node at (-0.3,3.1) {2};
\node at (-0.3,1.55) {1};
\node at (5,-0.3) {$r=1$};
\node (x_1) at (6.5,0) {$x_1$} ;
\node (cost) at (0,3.8) {cost} ;

\draw[->] (0,0) -- (6.2,0) ;
\draw[->] (0,0) -- (0,3.5) ;
\draw (0,3) -- (5,3) ;
\draw[dotted] (5,0) to (5,3.2) ;
\draw[dotted] (0,1.5) to (5,1.5) ;
\draw[dashed] (0,3) to (5,1.5) ;
\node at (5.5,1.55) {$C_1$} ;
\node at (5.5,3) {$C_2$} ;
\end{tikzpicture}\vspace*{-2mm}
\caption{With $n=2$ and $r=1$ (so that $x_2=1-x_1$), consider $\alpha_{1,2}=\alpha_{2,1}=1$, $f_1(x_1)=1$ and $f_2(x_2)=1+x_2$.
This gives $C_1(\x)=2-x_1$, $C_2(\x)=2$, and $\Phi(\x)=x_1+(x_2+\frac{1}{2}x_2^2)+x_1x_2=(3-x_1^2)/2$ (this potential function is given by Equation~\eqref{eq:formulePhi} in the upcoming Proposition~\ref{prop:potentialFormula}; we can  check that this is indeed a potential. Otherwise, other potentials are equal to this one up to an additive constant; also note that a potential can be expressed with $x_1$ alone, with $x_2$ alone, or with both $x_1$ and $x_2$). There are two equilibria, one in $x_1=1$ which is a minimum of the potential function, and one
in $x_1=0$, which is a maximum.
}
\label{fig:eq_vs_phi}
\end{figure}
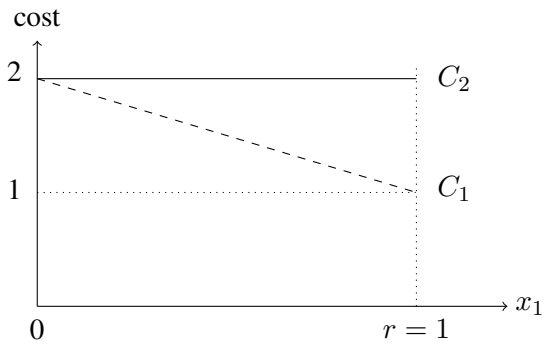

Note that since $\Phi$ is continuous on the compact set $\Delta_r(n)$, it admits a
global minimum on $\Delta_r(n)$, which gives another, more constructive, proof
of the existence of an equilibrium  when there is a potential function.

\smallskip

\subsection{Potential function and graphical NBGs} \label{sub:whygraphical}

\noindent
To build examples where there is a potential function, one could consider a differentiable function $\Phi$ and define $C_i(\x)$ as $\frac{\partial \Phi(\x)}{\partial x_i}$.

\medskip
However, to identify classes of games where there always exists a potential function, we can restrict our attention to cost functions which are decomposable in the following natural way:
\[C_i(\x) = f_i(x_i) + \sum_{j \neq i} f_{j,i} (x_j) \]
where the $f_i$'s and $f_{i,j}$'s are continuously differentiable functions (for
practical reasons).

\medskip
Supposing that the game
admits a potential function $\Phi$, by Schwarz's theorem, we must
have for all $i \neq j$
\[ \frac{\partial^2 \Phi}{\partial x_i \partial x_j}(\x)= \frac{\partial^2 \Phi}{\partial x_j \partial x_i}(\x),\]
so
\[\frac{\partial f_{j,i}(x_j)}{\partial x_j} = \frac{\partial
  f_{i,j}(x_i)}{\partial x_i}.\]
Therefore, these functions must be linear, of the form $f_{j,i} = c x_j$
and $f_{i,j} = c x_i$ for a $c \in \R$; whence the definition of
symmetric graphical games.

\medskip
The following is then easy to check.

\begin{proposition} \label{prop:potentialFormula}
 Let $G=(n,r,f,\alpha)$ be a symmetric graphical game. Define:
 \[ \Phi : \Delta_r(n) \longrightarrow \R\]
 by
 \begin{equation} \label{eq:formulePhi}
 \Phi(\x) = \sum_i \int_0^{x_i} f_i(t)dt + \sum_{i < j} \alpha_{i,j} x_i x_j.
 \end{equation}
 Then $\Phi$ is a potential function for $G$.
\end{proposition}

\noindent {\bf Remark.} A graphical NBG which is not symmetric may admit no potential. Indeed, by the discussion above, if $\alpha_{i,j} \neq \alpha_{j,i}$ and the $f_i$'s
are continuously differentiable, then
we see by Schwarz's theorem that a potential function cannot exist.

\medskip
In the following, we investigate how our model of NBG leads to equilibria that can be good or bad from a social viewpoint.

\subsection{Social costs}

A social cost is a way to average all the costs paid by the infinitesimal players in the game and quantify by a single number the cost, hence the quality of equilibria, if we want to compare them.

\eject
\noindent We define two social costs:
\begin{itemize}
\item The {\it utilitarian social cost}, defined as
 \[ {\mathcal C}_u(\x) = \frac1{r} \sum_{i \in [n]} x_i C_i(\x), \]
is simply the average cost of an infinitesimal mass in the graph. It can be low even if a small fraction of the mass pays a high cost.

\item The {\it egalitarian social cost}, defined as
 \[ {\mathcal C}_e (\x) = \max_{i \in [n], x_i >0} C_i(\x)\]
 is the maximum cost encountered among vertices with a positive quantity of mass. Minimizing this cost is more complicated than minimizing the utilitarian cost since it considers all infinitesimal players.
\end{itemize}
We should note that for all $\x \in \Delta_r(n)$, we have ${\mathcal C}_u(\x) \leq {\mathcal C}_e(\x)$, and that there is equality if $\x$ is an equilibrium.

\subsection{A Braess-like Paradox}

Braess’s Paradox was discovered by Braess~\cite{braess1968paradoxon}
in the context of nonatomic selfish routing. In this model of routing a continuous mass of vehicles from a source to a destination, one shows the paradoxical phenomenon that opening a new road leads to a dramatic increase in the social cost of the (unique) equilibrium of the model.

Despite being even more straightforward than the nonatomic selfish routing model, we show that our NBG model admits examples that lead to similar conclusions.

\medskip
Consider the following example with two vertices:
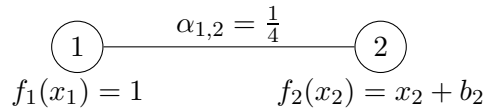
\begin{center}
 \begin{tikzpicture}
  \node[draw,circle] (a) at (0,0) {$1$};
  \node[draw,circle] (b) at (4,0) {$2$} ;
  \draw (a) edge (b) ;
  \node (t) at (2,.3) {$\alpha_{1,2} = \frac{1}4$} ;
  \node (f1) at (0,-.6) {$f_{1}(x_1) = 1$} ;
  \node (f2) at (4,-.6) {$f_2(x_2) = x_2 + b_2$} ;
 \end{tikzpicture}
\end{center}
with total mass $r=1$. We can write a mass distribution as $\x=(x_1,1-x_1)$ and
costs, as a function of $x_1$, are:
\[C_1(\x) = 1 + \frac1{4} (1-x_1) = - \frac1{4} x_1 + \frac5{4}\]
and
\[C_2(\x) = 1-x_1 + b_2 + \frac1{4}x_1 = -\frac3{4} x_1 + 1 + b_2 .\]
If $b_2 \in [\frac{1}4, \frac3{4}]$, costs intersect in $2b_2 - \frac{1}2$ and
there are no other equilibria (see Figure~\ref{ALfig:ex_eq2AL}). The common value of both costs is then $\frac{11}8 - \frac{b_2}2$, which grows when $b_2$ decreases; it is also the social cost (utilitarian or egalitarian) in this equilibrium. This means that a lower vertex-cost function can lead  to a
worse equilibrium in the sense of both social costs.

\begin{figure}[ht]
%%\vspace*{-2mm}
\centering
%%\scalebox{0.95}{
\begin{tikzpicture}
%%% b=1/4
\node at (-2.5,-0.8) {$b_2=1/4$} ;
\node (origine) at (-4.0,-0.3) {0};
\node at (-2,-0.3) {$r=1$};
%\node at (-3,-0.3) {$1/2$};
\node (x_1) at (-1.3,0) {$x_1$} ;
\node (cost) at (-4,3.7) {cost} ;
\node at (-4.5,0.7) {1/2};
%\node at (-4.5,1.3) {3/4};
\node at (-4.3,1.9) {1};
%\node at (-4.5,2.2) {9/8};
\node at (-4.5,2.5) {5/4};
%\node at (-4.5,3.1) {3/2};
%\node at (-4.5,3.7) {7/4};
\draw[->] (-4,0) -- (-1.6,0) ;
\draw[->] (-4,0) -- (-4,3.4) ;
%\draw[dotted] (-4,2.2) -- (-3,2.2) ;
\draw[dotted] (-2.0,0) -- (-2.0,2.4) ;
\draw[dotted] (-4,0.7) -- (-2.0,0.7) ;
%\draw[dotted] (-3,0) -- (-3,2.2) ;
\draw[dotted] (-4,1.9) -- (-2.0,1.9) ;
\draw (-4,2.5) -- (-2,0.7) ;
\draw[dashed] (-4,2.5) to (-2,1.9) ;
\node at (-1.5,1.9) {$C_1$} ;
\node at (-1.5,0.7) {$C_2$} ;

%%% b=1/2
\node at (1.5,-0.8) {$b_2=1/2$} ;
\node (origine) at (0,-0.3) {0};
\node at (2,-0.3) {$r=1$};
\node at (1,-0.3) {$1/2$};
\node (x_1) at (2.7,0) {$x_1$} ;
\node (cost) at (0,4.2) {cost} ;
%\node at (-0.5,0.7) {1/2};
\node at (-0.5,1.3) {3/4};
\node at (-0.3,1.9) {1};
\node at (-0.5,2.2) {9/8};
\node at (-0.5,2.5) {5/4};
\node at (-0.5,3.1) {3/2};
%\node at (-0.5,3.7) {7/4};
\draw[->] (0,0) -- (2.4,0) ;
\draw[->] (0,0) -- (0,3.9) ;
\draw[dotted] (0,2.2) -- (1,2.2) ;
\draw[dotted] (2.0,0) -- (2.0,2.4) ;
\draw[dotted] (0,1.3) -- (2.0,1.3) ;
\draw[dotted] (1,0) -- (1,2.2) ;
\draw[dotted] (0,1.9) -- (2.0,1.9) ;
\draw (0,3.1) -- (2,1.3) ;
\draw[dashed] (0,2.5) to (2,1.9) ;
\node at (2.5,1.9) {$C_1$} ;
\node at (2.5,1.3) {$C_2$} ;
%%% b=3/4
\node at (5.5,-0.8) {$b_2=3/4$} ;
\node (origine) at (4.0,-0.3) {0};
\node at (6,-0.3) {$r=1$};
%\node at (5,-0.3) {$1/2$};
\node (x_1) at (6.7,0) {$x_1$} ;
\node (cost) at (4,4.7) {cost} ;
%\node at (3.5,0.7) {1/2};
%\node at (3.5,1.3) {3/4};
\node at (3.7,1.9) {1};
%\node at (3.5,2.2) {9/8};
\node at (3.5,2.5) {5/4};
%\node at (3.5,3.1) {3/2};
\node at (3.5,3.7) {7/4};
\draw[->] (4,0) -- (6.4,0) ;
\draw[->] (4,0) -- (4,4.4) ;
%\draw[dotted] (4,2.2) -- (5,2.2) ;
\draw[dotted] (6.0,0) -- (6.0,2.4) ;
%\draw[dotted] (4,1.3) -- (6.0,1.3) ;
%\draw[dotted] (5,0) -- (5,2.2) ;
\draw[dotted] (4,1.9) -- (6.0,1.9) ;
\draw (4,3.7) -- (6,1.9) ;
\draw[dashed] (4,2.5) to (6,1.9) ;
\node at (5.5,2.8) {$C_2$} ;
\node at (5.7,1.6) {$C_1$} ;
\end{tikzpicture}%% \vspace*{-2mm}
\caption{The costs and equilibria for three different values of $b_2$.}
\label{ALfig:ex_eq2AL}
\end{figure}
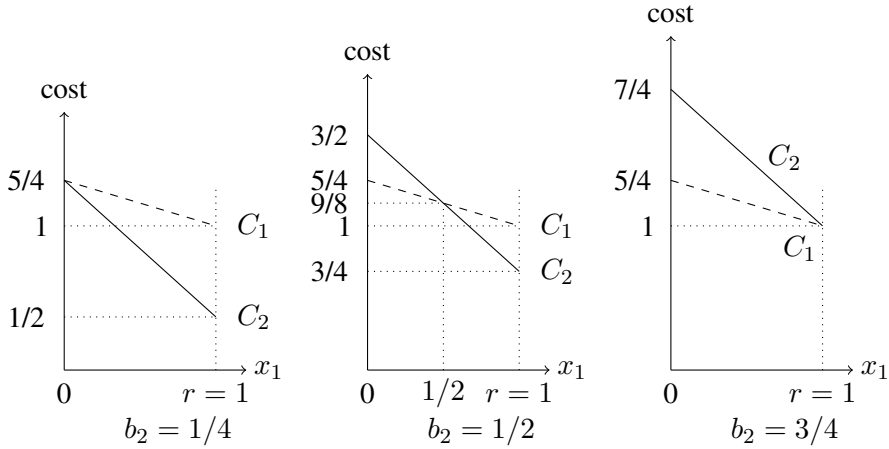

\medskip
Note that for each mass distribution $\x$, social costs ${\mathcal C}_u(\x)$ and ${\mathcal C}_e(\x)$ always decrease when we reduce $b_2$ or, more generally, vertex-cost functions; here, what increases when we reduce $b_2$ is the social cost of the unique equilibrium. Also, note that we obtained this paradox in the simple context of a symmetric affine NBG. In this case, a potential function with a unique minimum exists (so it is not a protection against these paradoxes).

\subsection{The Price of Anarchy}\label{sec:qualite:1}

The concept of a {\it price of anarchy} (PoA) is a popular measure for the inefficiency of equilibria in games. If there are multiple equilibria, we consider the worst case and quantify the cost of this worst equilibrium regarding the best configuration of the game, i.e. an optimal outcome where we do not have the constraints of equilibria.

\medskip
Let $Eq\subseteq \Delta _r(n)$ denote the set of equilibria of an NBG $G=(n,r,C)$.
We define the {\it price of anarchy} for the utilitarian social cost, $PoA_u(G)$, as follows:
\begin{equation} \label{eq:formulePOanarAL}
PoA_u (G) = \frac{\displaystyle \max_{\x \in Eq} {\mathcal C}_u(\x)}{\displaystyle \min_{\x \in \Delta_r(n)} {\mathcal C}_u(\x)},
\end{equation}
and $PoA_e(G)$ is defined likewise for the egalitarian social cost.
Unfortunately, the existence of several equilibria in NBGs leads to the price of anarchy unbounded even in elementary classes of NBGs.

\begin{proposition} \label{prop:PoAaffine}
 The price of anarchy (with both social costs) is unbounded on normal $\alpha$-uniform NBGs.
\end{proposition}

\begin{proof}
 Consider the following example on two vertices, where %A%$0< \-\lambda < 1$ and
 $r=1$.

 
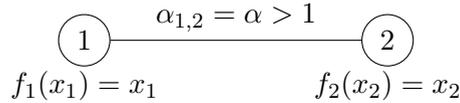
\begin{center}
  \begin{tikzpicture}
   \node[draw,circle] (a) at (0,0) {$1$};
   \node[draw,circle] (b) at (4,0) {$2$} ;
   \draw (a) edge (b) ;
   \node (t) at (2,.3) {$\alpha_{1,2} = \alpha > 1$} ;
   \node (f1) at (0,-.6) {$f_{1}(x_1) = x_1$} ;
   \node (f2) at (4,-.6) {$f_2(x_2) = x_2$} ;
  \end{tikzpicture}
 \end{center}

We can write both cost functions as functions of $x_1$
and obtain
 \[C_1(x_1) = (1-\alpha) x_1 + \alpha \]
 and
 \[C_2(x_1)= (\alpha - 1) x_1 + 1\]

 There are three equilibria:
 \begin{itemize}
   \item in $x_1 = 0$ we have $C_1(0) = \alpha > 1 = C_2(0) $ ; the social cost is 1;
   \item likewise in $x_1 = 1$, the social cost is 1 also. It is easy to see that this is a global optimum for the social cost;
   \item in $x_1 = 0.5$, we have $C_1(0.5) = C_2(0.5) = \frac{1+\alpha}{2}$, this common value being the social cost.
 \end{itemize}
 Hence, we see in this example that the price of anarchy is $\frac{1+\alpha}{2}$, but we can consider $\alpha$ as large as needed.
\end{proof}

We could derive bounds on the price of anarchy in simple cases (for normal linear NBGs, for instance) based on the size of the coefficients that appear. But instead, we shall now focus on another measure, more suited to the case of NBGs.

\subsection{The Price of Stability} \label{sec:qualite:2}

The {\it price of stability} (PoS) of the game, for the utilitarian social cost, is defined as
\begin{equation} \label{eq:formulePOSAL}
PoS_u (G) = \frac{\displaystyle \min_{\x \in Eq} {\mathcal C}_u(\x)}{\displaystyle \min_{\x \in \Delta_r(n)} {\mathcal C}_u(\x)} .
\end{equation}
The price of stability of the egalitarian social cost, $PoS_e(G)$, is defined similarly, using ${\mathcal C}_e$. The so-called stability corresponds to the fact that we require the mass distribution to be an equilibrium, hence stable concerning unilateral deviation of infinitesimal players; hence the price of stability is a measure of the increase of global social cost due to this stability.

\medskip
Note that for both social costs, we have $PoA(G) \geq PoS(G)\geq 1$. It is also easily noted that

\begin{equation} \label{eq:ineg_pos}
  PoS_e(G) \leq PoS_u(G).
\end{equation}
We now give upper bounds on $PoS(G)$ in different cases.
\begin{proposition} \label{prop:generalBoundPoS}
 Let $\cal C$ be a class of vertex-cost functions and $0 < \gamma \leq \frac1{2}$ such that, for all
 $f \in \cal C$, one has
 \[ \int_0^x f(t)dt \geq \gamma x f(x).\]
 Then, the price of stability for all graphical symmetric NBGs with cost functions in $\cal C$ is at most $\frac1{\gamma}$ for both social costs.
\end{proposition}
\noindent We prove the result for the utilitarian social cost, and the other case follows from (\ref{eq:ineg_pos}).
To prove Proposition~\ref{prop:generalBoundPoS}, we first prove the following lemma.
\begin{lemma} \label{lem:phiEncadre}
 With $\gamma$ as defined in the previous Proposition, and for all
 mass distributions $\x$ of a graphical symmetric NBG $G$ with potential function $\Phi$
 (as defined in Proposition~\ref{prop:potentialFormula}), one has
 \[ \gamma \cdot {\mathcal C}_u(\x) \leq \Phi(\x) \leq {\mathcal C}_u(\x).\]
\end{lemma}
\begin{proof}
In the case of a symmetric graphical game $(n,r,f,\alpha)$, the utilitarian social cost is equal to:

\begin{equation} \label{eq:formuleCu}
{\mathcal C}_u (\x) = \sum_i x_i f_i(x_i) + 2 \sum_{i<j} \alpha_{i,j} x_i x_j.
\end{equation}
Since vertex-cost functions $f_i$ are non-decreasing, the upper bound is clear from (\ref{eq:formulePhi}) and (\ref{eq:formuleCu}). The lower bound comes from the definition of $\gamma$ and $\gamma \leq \frac1{2}$.
\end{proof}
\noindent {\it Proof of Proposition~\ref{prop:generalBoundPoS}.} Let $\x_b$ be a best equilibrium, i.e. an equilibrium that minimizes the utilitarian social cost ${\mathcal C}_u$; let
 $\x_{\phi}$
 be a mass distribution that minimizes the potential function $\Phi$, which is an
 equilibrium
 by Proposition~\ref{prop:minPhi}, and let $\x^*$ be a mass distribution that
 minimizes ${\mathcal C}_u$, so that $PoS_u (G) = \frac{{\mathcal C}_u(\x_b)}{{\mathcal C}_u(\x^*)}$. Then we have
 \begin{align*}
  {\mathcal C}_u(\x_b) &\leq {\mathcal C}_u(\x_\phi)          \hskip 1cm \text{\small (since $\x_b$ minimizes ${\mathcal C}_u$ among equilibria)}\\
       &\leq \frac1{\gamma} \Phi(\x_\phi) \hskip 1cm \text{\small (by Lemma \ref{lem:phiEncadre})}\\
       &\leq \frac1{\gamma} \Phi(\x^*)  \hskip 1cm \text{\small (by definition of $\x_{\phi}$)}\\
       &\leq \frac1{\gamma} {\mathcal C}_u(\x^*)    \hskip 1cm \text{\small (by Lemma \ref{lem:phiEncadre}), }
 \end{align*}
 hence
 \[PoS_u (G) \leq \frac1{\gamma}.\]
\; \hfill $\QED$ %\end{proof}

Using the previous relation and Inequality~\eqref{eq:ineg_pos}, we can deduce that:
\begin{corollary} \label{AL007}
 On the class of polynomial functions with real nonnegative coefficients of degree
 $d \geq 1$,  both prices of stability  are at most $d+1$.
\end{corollary}
\noindent For $d=0$, the price of stability is at most 1, hence equal to 1. For $d=1$ (the affine case), the price of stability is between $1$ and $2$, and we can give a family of examples reaching asymptotically the bound $2$, see Figure~\ref{fig:boundAffineReached}.

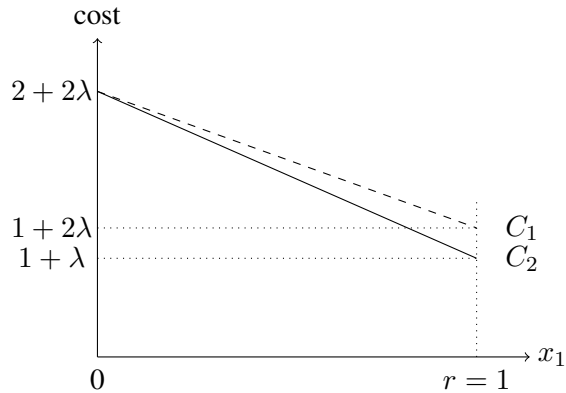
\begin{figure}[h]
\centering
\begin{tikzpicture}
\node (origine) at (0,0.2) {0};
\node at (5,0.2) {$r=1$};
\node (x_1) at (6,0.5) {$x_1$} ;
\node (cost) at (0,5) {cost} ;
\node at (-.6,4) {$\small 2+2\lambda$} ;
\node at (-.6,2.2) {$\small 1+2\lambda$} ;
\node at (-.6,1.8) {$\small 1+\lambda$} ;
\draw[->] (0,0.5) -- (5.7,0.5) ;
\draw[->] (0,0.5) -- (0,4.7) ;
\draw (0,4) -- (5,1.8) ;
\draw[dashed] (0,4) to (5,2.2) ;
\draw[dotted] (5,0.5) -- (5,2.6) ;
\draw[dotted] (0,2.2) -- (5,2.2) ;
\draw[dotted] (0,1.8) -- (5,1.8) ;
\node at (5.6,2.2) {$C_1$} ;
\node at (5.6,1.8) {$C_2$} ;
\end{tikzpicture}
\caption{A family of graphical games reaching the maximum price of stability of $2$ for
affine vertex-cost functions. With $n=2$ and $r=1$, consider $\alpha_{1,2}=\alpha_{2,1}=1$, $f_1(x_1)=1+2\lambda$ and $f_2(x_2)=(2+\lambda) x_2 + \lambda$.
This gives $C_1(\x)=2+2\lambda -x_1$ and $C_2(\x)=2+2\lambda-(1+\lambda)x_1$.
The only equilibrium of the game has cost $2+2\lambda$, while the best utilitarian social cost is obtained for $x_1=1$ and is equal to $1+2\lambda$ (this social cost is decreasing on $[0,1]$ if $0 < \lambda < 1$). Thus, we can reach the bound of $2$ for $PoS(G)$ by letting $\lambda$ go to $0$.
}
\label{fig:boundAffineReached}
\end{figure}

\medskip
We can be more specific when all vertex-cost functions are linear: the utilitarian social cost and the potential functions (as defined by Equations~\eqref{eq:formulePhi} and~\eqref{eq:formuleCu}) are related by:
\[{\mathcal C}_u(\x) = 2 \Phi (\x).\]

Observe by Proposition~\ref{prop:minPhi} that there exists an equilibrium corresponding to  the minimum of $\Phi$. Since ${\mathcal C}_u$ and $\Phi$ are minimum on the same
mass distributions and using Inequality~\eqref{eq:ineg_pos}, we get

%Hence, the price of stability is 1, by Proposition~\ref{prop:minPhi} and since ${\mathcal C}_u$ and $\Phi$ are minimum on the same
%mass distributions.

\begin{corollary}

 \label{prop:stability:linear}
  On the class of linear functions with real nonnegative coefficients, both prices of stability are   $1$.
\end{corollary}

\section{The $\alpha$-uniform graphical NBG for some graphs} \label{sec:graph}
This section considers that the underlying undirected graph can be path, cycle and complete bipartite graph.  From the example of the path, we derive the embryo of an algorithm outputting an equilibrium, which could be adapted to general graphs. Using the linearity of the cost functions, without loss of generality, we can assume that the total mass is equal to~$1$ ($r=1$).

\medskip
We can observe that, in a $\alpha$-uniform NBG,

\begin{rules} \label{rule:1}
An equilibrium  cannot have an uncharged vertex with only uncharged neighbours.
\end{rules}

Indeed, the cost of an uncharged vertex with uncharged neighbours is zero, and the mass distribution cannot be an equilibrium because the costs of vertices are strictly positive.

\begin{rules} \label{rule:2} An equilibrium cannot have one uncharged vertex with only one charged neighbour, unless $\alpha \geq 1$.
\end{rules}

Indeed, the cost of an uncharged vertex with only one charged neighbour is less than that of the charged neighbour itself when  $\alpha < 1$, and the mass distribution cannot be an equilibrium by definition.

\begin{rules} \label{rule:3} If  an equilibrium $\x$  has an uncharged vertex $i$ with $k$ charged neighbours,  then $\alpha \geq \frac{1}{k}$. If $\alpha = \frac{1}{k}$, then all the neighbours of~$i$ have the same mass, and none of them has charged neighbours.
\end{rules}

Indeed, suppose that vertex $i$ is not charged in $\x$. The definition of equilibrium in $\alpha$-uniform   NBGs implies these inequalities: for any  neighbour $j$ of~$i$, $C_{j}(\x)\geq x_{j}$,  and $C_i(\x)=\alpha \sum_{\ell \in N[i]} x_{\ell} \geq C_{j}(\x)$ . These inequalities are valid when $\alpha \geq \frac{1}{k}$. %We can extend this argument to arbitrary graphs. More generally, in a graph $G$, if an uncharged vertex has $k$ charged neighbours, then $\alpha \geq \frac{1}{k}$.
If $\alpha = \frac{1}{k}$, then $C_i(\x)=\frac{1}{k} \sum_{\ell \in N[i]} x_{\ell}\geq x_j$ for all~$j\in N[i]$, i.e., the average of the $k$~neighbours masses is at least the mass of each neighbour, so all masses are equal. Moreover, $C_i(\x)=C_j(\x)$ for all~$j$, so no neighbour of~$i$ can receive an additional cost from a neighbour.

\begin{rules} \label{rule:4} Let $u$ and $v$ be two vertices having the same neighbourhood.  If  a mass distribution $\x$ is an equilibrium,  then $x_u = x_v$.
\end{rules}

Indeed, if we let $N=\alpha \sum_{\ell \in N[u]} x_{\ell}=\alpha \sum_{\ell \in N[v]} x_{\ell}$, then the definition of equilibrium in $\alpha$-uniform   NBGs implies the equalities  $C_{u}(\x)= x_u + N$ and   $C_{v}(\x)= x_v + N$. Then $x_u \neq x_v$, say $x_u>x_v$,  implies that $C_{u}(\x)>C_{v}(\x)$, so $x_u$ is uncharged,  contradicting $x_u>x_v$.

\subsection{The case when  the underlying  graph is a path}
%\antoine{{\bf Je reprendrai tout ça avec mes notes}}
In this subsection, we consider the $\alpha$-uniform case for different values of
$\alpha$, in a quite simple graph: the path with $n$ vertices, $P_n$.

Paths are seemingly harder to deal with than cycles; even in the uniform case, they provide different interesting cases: for instance, for given $n$ and $\alpha$, one may have infinitely many equilibria. In the following, we do not necessarily provide the complete computations.

\medskip
First, we can look for equilibria $\x=(x_1, \ldots, x_n)$ where all vertices are charged or, if uncharged, have the same cost as charged vertices. We shall call such equilibria {\it uniform-cost} or {\it simply uniform}. Hence: \[\forall i,j \in V, \quad C_i(\x) = C_j(\x).\]
Let $c$ denote this common cost.
To calculate this type of equilibrium, we shall consider it as a solution to the system of equations~\eqref{ALeqAL1} below. This allows having a system of equalities rather than
inequalities; on the other hand, there are more unknowns.  The first three equations represent the fact that the costs of the vertices are identical: the first and the third correspond to the vertices~$1$ and~$n$. The fourth one corresponds to the constraint that the total mass equals~$1$.

\medskip
We then have:
\begin{equation}
\begin{aligned} \label{ALeqAL1}
 x_1 + \alpha x_2 &= c \\
 \forall \, 1 < i < n, \; \; \alpha x_{i-1} + x_i + \alpha x_{i+1} &= c \\
 \alpha x_{n-1} + x_n &=c \\
 \sum_{i\in [n]} x_i &= 1 \\
 \forall \, 1 \leq i \leq n, \; \; x_i &\geq 0.
\end{aligned}
\end{equation}

If we denote   $M_{n,\alpha}$ the $(n+1)\times (n+1)$ matrix
\[
 \begin{pmatrix}
  1 & \alpha &  & & & & & -1 \\
  \alpha & 1 & \alpha & & & & & -1 \\
  & \ddots & \ddots & \ddots & & & & -1\\
  & & \alpha & 1 & \alpha & & & -1\\
  & & & \ddots & \ddots & \ddots & & -1\\
  & & & & & \alpha & 1 & -1 \\
  1 & 1 & \cdots & 1 & 1 & 1 & 1 & 0 \\
 \end{pmatrix},
\]
where all blank entries are $0$,  a solution of System~\eqref{ALeqAL1} above gives in
particular a solution to the linear system
\begin{equation} \label{eq:matrixsystem}
 M_{n,\alpha} \cdot (x,c) = \begin{pmatrix} 0 \\ 0 \\ \vdots \\ 0 \\ 1 \\ \end{pmatrix},
\end{equation}
which will help us to find an equilibrium in $P_n$ (we know that there exists at least one). The same is true for cycles, with only a slight modification of the matrix (see Section~\ref{sec:cycle}).

\medskip
The solutions corresponding to uniform-cost equilibria depend on the determinant of  matrix $M_{n,\alpha}$ and thus also on the value of $\alpha$.

\paragraph{Case when  $n=2$.} We have $det(M_{2,\alpha}) = - 2 \alpha +2.$

For $\alpha = 1$, we have $det(M_{2,\alpha}) = 0$ and any mass distribution is a uniform-cost equilibrium. For $\alpha \neq 1$, we have at most one uniform-cost equilibrium. For
reasons of symmetry this equilibrium is $(\frac1{2},\frac1{2})$.

 \paragraph{Case when  $n=3$.} By computation, we get  $det(M_{3,\alpha}) = -4\alpha +3.$
Hence, using the same argument as previously, there is at most one uniform-cost equilibrium for $\alpha \neq \frac{3}{4}$.
By computation, this equilibrium would be
\[\x_0^* = \left(\frac{\alpha-1}{4\alpha-3}, \frac{2\alpha-1}{4\alpha-3}, \frac{\alpha-1}{4\alpha-3}\right).\]
For this to be an equilibrium, its entries must be nonnegative,
which is the case if and only if $0 \leq \alpha \leq \frac{1}{2}$ or $\alpha \geq 1$.

For $\alpha = \frac3{4}$, System~\eqref{eq:matrixsystem} has no solution. Hence, there is no uniform-cost equilibrium.
We can look for non-uniform equilibria (with at least one uncharged vertex). Two cases are possible:
\begin{itemize}
\item if the middle vertex $2$ is uncharged, then the only equilibrium
 must be $\x_1^*=(\frac1{2},0,\frac1{2})$, and this is a valid equilibrium only if
 $\alpha \geq \frac1{2}$.
\item there is no equilibrium where vertex $1$ or $3$ is uncharged, say $x_1=0$, $x_2>0$, $x_3>0$, since the cost in vertex $1$ would be less than the cost in vertex $3$. The equilibrium $\x_2^*=(0,1,0)$ is possible only if $\alpha \geq 1$.
\end{itemize}
So, to sum up, there are three equilibria, $\x_0^*$, $\x_1^*$ and $\x_2^*$  if $\alpha > 1$. Otherwise, if $\alpha =1$, there are two equilibria, $\x_2^*\;(=\x_0^*)$ and $\x_1^*$. If $1> \alpha >\frac1{2}$, then $\x_1^*$ is the only equilibrium. If $\alpha =\frac1{2}$, the only equilibrium is $\x_1^*\;(=\x_0^*)$. And if $\frac1{2}> \alpha \geq 0$, then $\x_0^*$ is the only equilibrium.
\paragraph{ Case when $n=4$:} We have $det(M_{4,\alpha}) = 2(\alpha^2+\alpha-1)(\alpha-2)$,
whose nonnegative roots are $2$ and $\phi=\frac{\sqrt{5}-1}2 \approx 0.618$.

\medskip
Hence, if $\alpha \neq 2$ and $\alpha \neq \phi$, we have at most one uniform-cost equilibrium. The solution
to System~\eqref{eq:matrixsystem} is then $x_1=x_4=\frac{1}{4-2\alpha}$, $x_2=x_3=\frac{1-\alpha}{4-2\alpha}$, which yields nonnegative values, i.e., a uniform equilibrium when $0\leq \alpha \leq 1$. As mentioned above, a search for non-uniform equilibria can be easily done, but we prefer to focus on the case $\alpha =\phi$ because it leads to infinitely many equilibria. First, any quadruple $(x_1,x_2,x_3,x_4)$ such that
\begin{equation} \label{n=4a} x_1+x_4=\frac{1}{2-\phi}\approx 0.724,
\end{equation}
\begin{equation} \label{n=4b} x_2=\frac{-\phi}{2-\phi}+x_4(1+\phi)\approx -0.447 +1.618\,x_4,
\end{equation}
\begin{equation} \label{n=4c} x_3=\frac{1}{2-\phi}-x_4(1+\phi)\approx 0.724 -1.618\,x_4
\end{equation}
is a solution to  System~(\ref{eq:matrixsystem}). The nonnegativity condition then implies that $x_4 \in \;[\frac{\phi}{(\phi +1)(2-\phi)},$ $\frac{1}{(\phi +1)(2-\phi)}]$, i.e. $\approx 0.276 \leq x_4 \leq \approx 0.447$; then any triple $(x_1,x_2,x_3)$ verifying (\ref{n=4a})--(\ref{n=4c}) gives, together with $x_4$, a uniform equilibrium. All these solutions have the same cost, namely $\frac{1-\phi^2}{2-\phi}\approx 0.447$.
%\david{donner les détails ? y'a des gros nombres compliqués} %\antoine{{\bf oui, mais c'est rigolo (infinité de solutions)}}

\paragraph{ Case when $n=5$:}  We obtain
$det(M_{5,\alpha}) =(\alpha +1)(\alpha -1)(\alpha^2+8\alpha -5),$
whose nonnegative roots are $1$ and $\phi=\sqrt{21}-4 \approx 0.583$.

\medskip
Calculations show that there is no solution to  System~\eqref{eq:matrixsystem} when $\alpha =\phi$, implying that there is no uniform equilibrium. More calculations lead to the following results for $\alpha \in [0,1]$:
\begin{itemize}
\item if $0\leq \alpha \leq \frac{1}{2}$, there is a unique solution: a uniform equilibrium.
\item if $\frac{1}{2} \leq \alpha < 1$, then $\x_0^*=(\frac{1}{3},0,\frac{1}{3},0,\frac{1}{3})$ is the only equilibrium, and it is non-uniform unless $\alpha = \frac{1}{2}$.
\item if $\alpha =1$, $\x_0^*$ is an equilibrium, and any 5-tuple $(x_1,x_2,0,x_4,x_5)$ such that $x_1+x_2=x_4+x_5=\frac{1}{2}$ and $x_2+x_4\geq \frac{1}{2 }$ (e.g., $(\frac{1}{8},\frac{3}{8},0,\frac{1}{4},\frac{1}{4})$) is a non-uniform equilibrium, except the distributions $(x_1,x_2,0,x_1,x_2)$, which are uniform, with cost $x_1+x_2=\frac{1}{2}$.
\end{itemize}

%\david{j'ai perdu les détails, tu les as ?}
%\antoine{{\bf Oui - càd, au LRI}}an be found.

\paragraph{General case:}
\noindent We observe that the four previous determinants,  $det(M_{2,\alpha})$, $det(M_{3,\alpha})$, $det(M_{4,\alpha})$, $det(M_{5,\alpha})$, show no regularity, and it looks hard to find a general form for them, for their roots, and for the solutions to System~\eqref{eq:matrixsystem}. %We can, however, remark that the following rules, also valid for cycles, must be verified. %In an equilibrium with uncharged vertices,

\medskip
The first three Rules given above have the following implications when the graph is a path.

\medskip
Rule~\ref{rule:1} implies that an equilibrium  cannot have three or more consecutive uncharged vertices. %Moreover, when the underlying graph is a path,   we obtain a more constrained rule on uncharged vertices.
Rule~\ref{rule:2} implies that an equilibrium cannot have two consecutive uncharged vertices unless $\alpha \geq 1$.
Rule~\ref{rule:3} implies that no equilibrium has  uncharged vertices when $0\leq \alpha < \frac{1}{2}$.

\medskip
Because  such a game has at least one equilibrium by Theorem~\ref{th:existence:g},  we have that
\begin{itemize}
\itemsep=0.95pt
\item either $det(M_{n,\alpha})\neq 0$: the unique solution to System~\eqref{eq:matrixsystem} is nonnegative;
\item or $det(M_{n,\alpha})=0$: among the solutions to System~\eqref{eq:matrixsystem}, at least one is nonnegative. % (use $M^*_{n,\alpha}$ and System (\ref{eq:matrixsystem}$^*)$ for cycles).
\end{itemize}

\smallskip
From these three Rules and the above results on small values of $n$, we can conjecture that:
\begin{conj} \label{conjalph}
For $n\geq 2$ and $\alpha < \frac{1}{2}$, we have $det(M_{n,\alpha})\neq 0$; therefore, there is a unique equilibrium (which has a uniform cost) in $P_n$.
 \end{conj}
 \begin{example} In path $P_6$ with $\alpha =\frac{1}{4}$, the unique equilibrium is $\frac{1}{76}(15,11,12,12,11,15)$, with cost $\frac{71}{304}\approx .23$; for $\alpha =\frac{1}{3}$, it is $\frac{1}{38}(8,5,6,6,5,8)$, with cost $\frac{29}{114}\approx .25$. In $P_7$ with $\alpha =\frac{1}{4}$, it is $\frac{1}{240}(41,30,33,32,33,30,41)$, with cost $\frac{97}{480}\approx .20$; and for $\alpha =\frac{1}{3}$, it is $\frac{1}{71}(13,8,10,9,10,8,13)$, with cost $\frac{47}{213}\approx .22$.

% For $C_n$, the unique solution would obviously be $\frac{1}{n}(1,1, \ldots ,1)$.
 \end{example}

 \medskip

\noindent The above rules lead to
a starting point for an algorithm finding an equilibrium in a path with $n$ vertices, which, however, has too high a worst-case complexity due to Step $2$:
\begin{itemize}
\item Step $1$: compute $det(M_{n,\alpha})$.

(a) if $det(M_{n,\alpha})\neq 0$, there is a unique solution to System~\eqref{eq:matrixsystem}.
 \begin{itemize}
 \item if the solution is nonnegative, we have a uniform equilibrium. We know that this is always the case if $\alpha < \frac{1}{2}$.
 \item if the solution is negative,  we must search for non-uniform equilibria (see Step 2).
 \end{itemize}
 (b) if $det(M_{n,\alpha})= 0$, either there is no solution to System~\eqref{eq:matrixsystem}, or infinitely many. In the latter case, if some of these solutions are nonnegative,  we have obtained uniform equilibria (see the example when $n=4$). Otherwise, we are back to the search for non-uniform equilibria.
\end{itemize}

\begin{itemize}
\item Step $2$: search for a non-uniform equilibrium: try with uncharged vertices, respecting Rules~\ref{rule:1} and~\ref{rule:2} above (Rule~\ref{rule:3} is respected because we are in a case when no uniform equilibrium has been found).
 \end{itemize}
 This sketch of an algorithm can be adapted to general graphs, with modification of the determinant.  Note that uncharged vertices mean fewer unknowns but induce inequalities on the costs instead of equalities.

 \medskip
 Now, we restrict ourselves to the cases $\alpha = \frac{1}{2}$ and $\alpha =1$, the study of general $\alpha$ being seemingly out of reach. We recall Conjecture \ref{conjalph} when $\alpha <\frac{1}{2}$.

\paragraph{General case for $\alpha = \frac{1}{2}$.}

Both ends of the path must be charged. Rule~\ref{rule:3} above shows that if there is one uncharged vertex $x_i=0$, then $x_{i-1}=x_{i+1}$, and either the vertex $i-1$ or $i+1$ is an end, or $x_{i-2}=x_{i+2}=0$; similarly, $x_{i-3}=x_{i+3}=x_{i-1}$, and so on. There are two cases.

If $n$ is odd, then $(\frac{2}{n+1},0,\frac{2}{n+1},0,\frac{2}{n+1}, \ldots , 0,\frac{2}{n+1})$ is the only equilibrium, with uniform cost equal to $\frac{2}{n+1}$.

\medskip
If $n$ is even, $n=2q$, then all vertices are charged, and the only equilibrium is given by
\begin{align*}
 x_{2k} = kx_2, \; 1\leq k \leq q\\
 x_{2k+1} = x_1 - kx_2, \; 0\leq k \leq q-1\\
 x_1=\frac{1}{q+1}\\
 x_2= \frac{1}{q(q+1)},
\end{align*}
with uniform cost equal to $\frac{2q+1}{2q(q+1)}$. For instance, the equilibrium for $n=10$ is $$(\frac{1}{6},\frac{1}{30},\frac{2}{15},\frac{1}{15},\frac{1}{10},\frac{1}{10},\frac{1}{15},\frac{2}{15},\frac{1}{30},\frac{1}{6})=\frac{1}{30}(5,1,4,2,3,3,2,4,1,5);$$ its uniform cost is $\frac{11}{60}$. Observe that the odd positions are $5,4,3,2,1$ and the even positions are $1,2,3,4,5$.

\paragraph{General case for $\alpha = 1$.}
\begin{remark} \label{rmkalph}
Let $\x= (x_1,\dots, x_n)$ be an equilibrium.  When $\alpha =1$, a simple observation is that (in paths as well as in cycles) if $x_i=0$, $x_{i+1}>0$ and $x_{i+2}>0$, then $x_{i+3}=0$.
\end{remark}
The search for uniform-cost solutions is straightforward (no need to compute the determinant) and divides into three cases.
\begin{itemize}
    \item If $n=3k+2$, there is a simple infinity of solutions $(x_1,x_2,0,x_1,x_2,0, \ldots , x_1,x_2,0,x_1,x_2)$, with $x_1\geq 0$, $x_2\geq 0$, and $x_1+x_2=\frac{3}{n+1}$. This includes $(\frac{3}{n+1},0,0, \frac{3}{n+1},0,0,\ldots , \frac{3}{n+1},0)$ and its symmetric $(0,\frac{3}{n+1},0,0,\frac{3}{n+1},0, \ldots ,0,\frac{3}{n+1})$. The cost for all vertices is $x_1+x_2=\frac{3}{n+1}$. We can see that this means that $\alpha -1$ divides $det(M_{n,\alpha})$.
    \item When $n=3k+1$ or $n=3k$, there is a unique solution, respectively \\
     $(\frac{3}{n+2},0,0,\frac{3}{n+2},0,0, \ldots , \frac{3}{n+2})$ or $(0,\frac{3}{n},0,0,\frac{3}{n},0, \ldots ,0,\frac{3}{n},0)$, which means that in both cases $det(M_{n,1})\neq 0$.
\end{itemize}

Still for $\alpha =1$, the study of non-uniform-cost equilibria is similar to that for cycles. See below.

\subsection{The case  when  the underlying graph is a cycle} \label{sec:cycle}

We apply here the same argument  as for paths: we build a new matrix

\[ M^*_{n,\alpha} =\begin{pmatrix}
  1 & \alpha &  & & & & \alpha & -1 \\
  \alpha & 1 & \alpha & & & & & -1 \\
  & \ddots & \ddots & \ddots & & & & -1\\
  & & \alpha & 1 & \alpha & & & -1\\
  & & & \ddots & \ddots & \ddots & & -1\\
  \alpha & & & & & \alpha & 1 & -1 \\
  1 & 1 & \cdots & 1 & 1 & 1 & 1 & 0 \\
 \end{pmatrix}
\]
and its corresponding linear system (\ref{eq:matrixsystem}$^*).$

\medskip
Cycles and their determinants are more regular than paths, but no general pattern has been found for their determinants, roots, or solutions to (\ref{eq:matrixsystem}$^*)$.

\medskip
Rules~\ref{rule:1},~\ref{rule:2} and~\ref{rule:3}, as well as the sketch of the algorithm above, can be adapted to cycles.   Conjecture~\ref{conjalph}  becomes:

\begin{conj} \label{conjalph-1}
For $n\geq 2$ and $\alpha < \frac{1}{2}$, we have $det(M^*_{n,\alpha})\neq 0$; therefore there is a unique equilibrium (which has a uniform cost) in $C_n$.
 \end{conj}
Unlike what happens for paths, we can observe that in $C_n$, the mass distribution described by $x_i=\frac{1}{n}$, $1\leq i \leq n$, is always a uniform-cost equilibrium. We sketch some results for $\alpha =\frac{1}{2}$ and $\alpha =1$.

\paragraph{General case for $\alpha =\frac{1}{2}$.}

There is a simple infinity of equilibria with uniform cost when $n$ is even, given by $(x_1,x_2, \ldots ,x_1,x_2)$, with $x_1\geq 0$, $x_2\geq 0$, $x_1+x_2=\frac{2}{n}$. This includes $(0,\frac{2}{n}, \ldots ,0,\frac{2}{n})$ and $(\frac{1}{n}, \ldots ,\frac{1}{n})$. The cost of each vertex is $x_1+x_2=\frac{2}{n}$. When $n$ is odd, the unique solution is with $x_i=\frac{1}{n}$.

\paragraph{General case for $\alpha = 1$.}

As for paths, it is straightforward to see that:

\smallskip
-- If $n=3k$, there is a double infinity of solutions, $(x_1,x_2,x_3, \ldots ,x_1,x_2,x_3)$, with $x_1 \geq 0$, $x_2 \geq 0$, $x_3 \geq 0$ and $x_1+x_2+x_3=\frac{3}{n}$. This means that $(\alpha -1)^2$ divides $det(M^*_{n,\alpha})$.

\smallskip
-- If $n=3k+1$ or $3k+2$, the only solution is with $x_i=\frac{1}{n}$.

\medskip

\noindent Using Remark \ref{rmkalph}, one can see that non-uniform solutions are combinations of configurations of the type
$$(\ldots 0,x,0,0,x,0,x,0,ax,(1-a)x,0,bx,(1-b)x,0, \ldots),$$
with $0<a<1$, $a\leq b <1$, and the sum of charges equal to $1$.

\subsection{The case  when  the underlying graph is a  complete bipartite graph}
%\antoine{A FAIRE d'après mes notes}
We consider the complete bipartite graph $K_{p,q}$ with $p+q=n$, $p\geq q$, and a mass distribution $\x=(x_1, \ldots ,x_{p}, x_{p+1}, \dots , x_{p+q})$ with $r=1$.

\medskip
Using Rule~\ref{rule:4} given above, we have $x_i=x_j$ for $1\leq i<j\leq p$ and $x_k=x_{\ell}$ for $p+1\leq k<\ell \leq p+q$. Let $a=x_1$ and $b=x_{p+1}$.

\medskip
\noindent We start with the search for equilibria with uncharged vertices. We may assume that either $a=0$ (and $b=\frac{1}{q}$) or $b=0$ (and $a=\frac{1}{p}$).
We can apply Rule~\ref{rule:3} to determine  when the equilibrium exists according to the value $\alpha$.
In the former case,  the equilibrium with $a=0$ exists when $\alpha \geq \frac{1}{q}$. Similarly in the latter case,
 the equilibrium with $b=0$ exists when $\alpha \geq \frac{1}{p}$.

\medskip
 We see that if none of these two conditions on $\alpha$ is fulfilled, i.e., $\alpha < \frac{1}{p}$, then, since an equilibrium always exists, we must have a uniform equilibrium with no uncharged vertices.

\medskip
\noindent Now we search for a uniform equilibrium with no uncharged vertices. By definition, we have $a+\alpha qb=b+\alpha pa$, i.e., $a(1-\alpha p)=b(1-\alpha q)$, and $pa+qb=r=1$. Since $a\neq 0$ and $b\neq 0$, we can consider two cases: (i) $1-\alpha p=1-\alpha q=0$, (ii) both $1-\alpha p$ and $1-\alpha q$ are nonzero.

\medskip
(i) If $1-\alpha p=1-\alpha q=0$, then $p=q$ and the only condition remaining is $a+b=\frac{1}{p}\,$: there is an infinity of equilibria, given by $(a,b)=(a,\frac{1}{p}-a)$, $0\leq a\leq \frac{1}{p}$, and the common cost is $\frac{1}{p}$.

\smallskip
(ii) If both $1-\alpha p$ and $1-\alpha q$ are nonzero, then $a=b\frac{1-\alpha q}{1-\alpha p}$.

\medskip
\indent \indent (1) $\frac{1}{p}<\alpha <\frac{1}{q}$. Then $a$ and $b$ have different signs, and no uniform equilibrium with no uncharged vertices exists.

\medskip
\indent \indent (2) $\alpha < \frac{1}{p}$ or $\alpha > \frac{1}{q}$. If $p+q-2\alpha pq$ has the same sign as $1-\alpha p$ and $1-\alpha q$, then
$$a=\frac{1-\alpha q}{p+q-2\alpha pq}, \; b=\frac{1-\alpha p}{p+q-2\alpha pq}$$
gives the only uniform equilibrium, with common cost $\frac{1-\alpha^2 pq}{p+q-2\alpha pq}$. Otherwise, we have to search for a non-uniform equilibrium.

\medskip

Note that the case $p=q=1$ gives the path $P_2$, $p=q=2$ gives the cycle $C_4$, and $q=1$ gives the star, see below.

\paragraph{Particular case of the star.} The graph now is $K_{n-1,1}$.

\smallskip
(1) If $x_n=0$, i.e., only the center is uncharged, then $x_i=\frac{1}{n-1}$, $1\leq i \leq n-1$, and necessarily $\alpha \geq \frac{1}{n-1}$. In this case, we have an equilibrium, which is uniform only for $\alpha =\frac{1}{n-1}$.

\smallskip
(2) If all leaves are uncharged, then $x_n=1$. This leads to an equilibrium whenever $\alpha \geq 1$, and this equilibrium is uniform for $\alpha =1.$

\smallskip
(3) If there is no uncharged vertex, then $b(1-\alpha)=a(1-(n-1)\alpha)$ and $b+(n-1)a=1$. If $\alpha =1$, then $a=0$ and $b=1$. If $\alpha =\frac{1}{n-1}$, then $b=0$ and $a=\frac{1}{n-1}$. These are two uniform equilibria but with uncharged vertices.

\medskip
If $\alpha <1$ and $\alpha > \frac{1}{n-1}$, then $a$ and $b$ have opposite signs and no uniform equilibrium exists.

\smallskip
If $\alpha >1$ or $\alpha < \frac{1}{n-1}$, then $a$ has the same sign as $b$. We obtain $a=\frac{1-\alpha}{n-2\alpha (n-1)}$ and $b=\frac{1-(n-1)\alpha}{n-2\alpha(n-1)}$. Both numbers are positive and yield a uniform equilibrium.

\medskip
Summarizing for the star: if $\alpha < \frac{1}{n-1}$, there is one uniform equilibrium, with all vertices charged. For all $\alpha \geq \frac{1}{n-1}$, there is one equilibrium with the centre uncharged, which is uniform for $\alpha = \frac{1}{n-1}.$ If $\alpha = 1$, there is one uniform equilibrium with  all  uncharged leaves. If $\alpha >1$, there is one uniform equilibrium with all charged vertices  and one non-uniform equilibrium with all uncharged leaves.

\end{document}